\documentclass[10pt]{article}
\usepackage[a4paper]{geometry}
\usepackage[utf8]{inputenc}
\usepackage[T1]{fontenc}
\usepackage{setspace,authblk}
\usepackage{amsmath,mathrsfs,amsfonts,amsthm,amssymb}

\usepackage[english]{babel}

\usepackage{graphicx}
 
\usepackage{hyperref}

\title{\singlespace {Revisiting the state-space model of unawareness}}

\author[]{Alex A.T. Rathke\thanks{NECCT/FEA-RP/USP, University of S\~ao Paulo. \texttt{alex.rathke@alumni.usp.br}}}

\date{\today}

\theoremstyle{plain}
\newtheorem{theorem}{Theorem}
\newtheorem*{theorem*}{Theorem}
\newtheorem{proposition}{Proposition}
\newtheorem*{proposition*}{Proposition}

\newtheorem*{remark*}{Remark}

\newtheorem*{condition*}{Condition}
\newtheorem{definition}{Definition}
\newtheorem*{definition*}{Definition}
\newtheorem{assumption}{Assumption}
\newtheorem*{assumption*}{Assumption}

\begin{document}

\maketitle

\begin{abstract} 

We propose a knowledge operator based on the agent's possibility correspondence which preserves her non-trivial unawareness within the standard state-space model. Our approach may provide a solution to the classical impossibility result that 'an unaware agent must be aware of everything'.

\end{abstract}


\section{Introduction} \label{Introduction}

The standard state-space model of knowledge and unawareness originates from the seminal work of \cite{aumann1976} on common knowledge, where it assumes a set of states $\Omega$, and $P_{1}$ is a partition of $\Omega$, representing the information structure of the agent. When a state $\omega \in \Omega$ obtains, the agent perceives only the partitional subset $P_{1}(\omega) \subseteq P_{1}$ which contains the true state $\omega$. For any event $E \subseteq \Omega$, to say the agent 'knows' $E$ means that $E$ includes at least one partitional subset $P_{1}(\omega)$ \cite{aumann1976}. 

The partitional structure in \cite{aumann1976} implies two key conditions. First, for any states $\omega, \omega^{\prime} \in P_{1}(\omega)$, it implies $P_{1}(\omega) = P_{1}(\omega^{\prime})$, i.e. the states $\omega, \omega^{\prime}$ are not distinguishable by the agent \cite{geanakoplos1989,samet1990}. Second, each partitional subset $P_{1}(\omega) \subseteq P_{1}$ includes all the states $\omega$ which defines it, i.e. $\omega \in P_{1}(\omega)$ for all $\omega \in \Omega$. Therefore, the concept of knowledge of an event $E$ based on the partition $P_{1}$ derives a set-theoretic definition of knowledge equal to $\{ \omega \in \Omega : P_{1}(\omega) \subseteq E \}$ \cite{bacharach1985,chen2012,dekel1998,fukuda2021,galanis2013,geanakoplos1989,heifetz2006,li2009,schipper2014,tada2024}, for this becomes a standard definition in the literature.

Imperfect information is represented by coarser partitions of $\Omega$ \cite{aumann1976}, while further limitations of knowledge require more generalised structures. \cite{geanakoplos1989,samet1990} propose the application of possibility correspondences $P : \Omega \rightarrow 2^{\Omega}$ mapping states $\omega \in \Omega$ to subsets of $\Omega$, which allow for general cases as e.g. $\omega \not\in P(\omega)$ for some state $\omega \in \Omega$. In special, \cite{geanakoplos1989} proposes the application of non-partitional possibility correspondences to represent the concept of unawareness, which refers to several levels of lack of knowledge about an event \cite{bacharach1985,chen2012,dekel1998,fukuda2021,geanakoplos1989,heifetz2006,li2009,schipper2014,tada2024}. Non-partitional possibility correspondences imply $P(\omega) \cap P(\omega^{\prime}) \neq \emptyset$ for two different possibility sets $P(\omega) \neq P(\omega^{\prime})$. In this approach, the agent is unaware of some event $E \subseteq \Omega$ at some state $\omega \in \Omega$ at a necessary condition that two possibility sets not included in $E$ share the same state $\omega \in P(\omega) \cap P(\omega^{\prime})$.

The current state-space approach is convenient for its syntax-free structure and its plain implementation into economic decision models \cite{heifetz2006,li2009,galanis2013,geanakoplos1989}. Nonetheless, the seminal works of \cite{dekel1998,modica1994} demonstrate a major inconsistency within the model, that an unaware agent which satisfies some well-established properties of knowledge must indeed be aware of everything\footnote{The DLR impossibility result shows that an agent must be either fully aware or fully unaware of all possible states \cite{dekel1998}, see Section \ref{Standard model and the DLR impossibility result}.}. The classical result known as \emph{Dekel-Lipman-Rustichini (DLR) impossibility result} indicates that any standard state-space model based on possibility correspondences $P$ cannot appropriately represent unawareness \cite{dekel1998,modica1994,galanis2013,heifetz2006,li2009,schipper2014}.

Studies address the existing limitations by analysing which properties of knowledge are inconsistent across each other \cite{chen2012,fukuda2021,tada2024}, and by devising improved models which solve the DLR impossibility result, with highlights to the approach of multiple state-spaces \cite{fukuda2021,heifetz2006,li2009,galanis2013,schipper2014}. In this study, we revisit the standard model in Section \ref{Standard model and the DLR impossibility result}, then in Section \ref{One approach to preserve non-trivial unawareness} we propose one approach to preserve the non-trivial unawareness of agents within the standard state-space model based on possibility correspondences $P$.

\section{Standard model and the DLR impossibility result} \label{Standard model and the DLR impossibility result}

The standard model assumes a set of states $\Omega$, an algebra of events $2^{\Omega}$, and a possibility correspondence $P : \Omega \rightarrow 2^{\Omega}$ which maps states $\omega \in \Omega$ to subsets $P(\omega) \subseteq \Omega$. The set $P(\omega)$ is the set of states which the agent considers to be possible when state $\omega$ obtains \cite{dekel1998,geanakoplos1989}. The possibility correspondence $P$ represents the information processing capacity of the agent \cite{bacharach1985,dekel1998,geanakoplos1989,li2009,schipper2014}, and it applies to define the agent's knowledge with respect to each event $E \subseteq \Omega$. For all $E \subseteq \Omega$, the standard knowledge operator is defined by \cite{bacharach1985,chen2012,dekel1998,fukuda2021,galanis2013,geanakoplos1989,heifetz2006,li2009,schipper2014,tada2024}

\begin{equation} \label{k01}
K(E) := \{ \omega \in \Omega : P(\omega) \subseteq E \} = KE . \\
\end{equation}

Eq. \ref{k01} implies that the agent knows $E$ at the state $\omega$ if $P(\omega) \subseteq E$. $KE$ is the set of states in which the agent knows $E$, and $\neg KE = \Omega \setminus KE$ is the set of states in which she does not know $E$. Iterations of the knowledge operator $K$ model the agent's introspection, so if the agent knows $E$ at the state $\omega$, then she knows that she knows $E$, i.e. $KE \subseteq K(KE)$. 

On the other hand, a set of states $F \subseteq \Omega$ in which the agent is unaware of the event $E$ implies that at any state $\omega \in F$, the agent does not know $E$, and she does not know that she does not know $E$, and she does not know that she does not know that she does not know $E$, and so on. Iterations of the complement of the knowledge $\neg K$ as defined in Eq. \ref{k01} over $E$ define the unawareness operator satisfying \cite{chen2012,dekel1998,fukuda2021,galanis2013,geanakoplos1989,heifetz2006,li2009,schipper2014,tada2024} \newline

\begin{equation} \label{u01}
\begin{array}{rl}
U(E) &: \subseteq \neg KE \cap \neg K \neg KE \cap \neg K \neg K \neg KE \cap \dots \\
\\
&= \bigcap_{i = 1}^{\infty} (\neg K)^{i} (E) , \\
\end{array}
\end{equation}

\noindent where $(\neg K)^{i}$ indicates that the $\neg K$ operator iterates $i$ times. For any event $E \subseteq \Omega$, a non-empty set $U(E) = UE \neq \emptyset$ refers to the condition of \emph{non-trivial unawareness} \cite{chen2012,dekel1998,fukuda2021,galanis2013,tada2024}, for it indicates the states $\omega \in UE$ in which the agent is unaware of the event $E$.

The $K,U$ operators in Eq. \ref{k01} and \ref{u01} become the basis of the standard state-space model. Remark that if the possibility correspondence $P$ is a partition of the state-space $\Omega$, then the $U$ operator in Eq. \ref{u01} provides an empty set, so there is no unawareness at all. The seminal works of \cite{geanakoplos1989,samet1990} propose the application of non-partitional possibility correspondences $P$ to model non-trivial unawareness. In this approach, for some event $E$ and some unaware state $\omega \in UE$, there exists some state $\omega^{\prime} \in E$ such that $\omega^{\prime} \in P(\omega) \cap P(\omega^{\prime})$ \cite{dekel1998,geanakoplos1989}.

For example, let a state-space $\Omega = \{a,b,c \}$ and the non-partitional possibility correspondence $P(a) = \{a \}$, $P(b) = \{b \}$, $P(c) = \Omega$. Apply the $K,U$ operators in Eq. \ref{k01} and \ref{u01}. The agent knows the event $E = \{a \}$ at the state $K(\{a \}) = \{a \}$, while she is unaware of $E$ at states satisfying $U(\{a \}) \subseteq \{c \}$. In this example, non-trivial unawareness clearly implies $U(\{a \}) = \{c \}$.

Studies propose several properties of the standard $K,U$ operators, see \cite{bacharach1985,chen2012,dekel1998,fukuda2021,galanis2013,geanakoplos1989,heifetz2006,li2009,schipper2014,tada2024}. For all $E \subseteq \Omega$, regard the following: \newline

\noindent I. Necessitation: $K\Omega = \Omega$;

\noindent II. KU introspection: $K(UE) = \emptyset$;

\noindent III. AU introspection: $UE \subseteq U(UE)$. \newline

Necessitation means that the agent knows tautologies\footnote{Necessitation derives directly from the standard $K$ operator in Eq. \ref{k01}.}, KU introspection means that the agent does not know what she is unaware of, and AU introspection means that if the agent is unaware of some event, then she is unaware of being unaware of that event.

Despite the properties proposed in literature, the seminal works of \cite{dekel1998,modica1994} show that $K,U$ operators satisfying necessitation, KU introspection and AU introspection necessarily eliminate all non-trivial unawareness. The classical \emph{DLR impossibility result} demonstrates under conditions that for any non-trivial unawareness $UE \neq \emptyset$, it implies \cite{chen2012,fukuda2021,galanis2013,heifetz2006,li2009,schipper2014,tada2024}

\begin{equation} \label{dlr01}
\begin{array}{rl}
\emptyset \neq UE &\subseteq U(UE) \text{ (AU introspection)} \\
\\
&\subseteq \neg K \neg K (UE) \text{ (Eq. \ref{u01})} \\
\\
&= \neg K \Omega \text{ (KU introspection)} \\
\\
&= \emptyset , \text{ (necessitation)} \\
\\
\end{array}
\end{equation}

\noindent therefore the contradiction $\emptyset \neq UE = \emptyset$. Recently, \cite{chen2012,fukuda2021,tada2024} demonstrate that non-trivial unawareness resulting from the $U$ operator in Eq. \ref{u01} is inconsistent with AU introspection, for it explains the DLR impossibility result. E.g. our example above with $\Omega = \{a,b,c \}$, $P(a) = \{a \}$, $P(b) = \{b \}$, $P(c) = \Omega$. If the agent is unaware of the event $E = \{a \}$ so we have $U(\{a \}) = \{c \}$, then we find $U(U(\{a \})) = U(\{c \}) = \emptyset$, therefore it violates AU introspection, $U(\{a \}) \not\subset U(U(\{a \}))$\footnote{Non-trivial unawareness $UE \neq \emptyset$ deriving from Eq. \ref{u01} also violates the property of 'negative introspection' defined as $\neg KE \subseteq K \neg KE$ \cite{chen2012,fukuda2021,tada2024}, e.g. in the same example with $\Omega = \{a,b,c \}$, $P(a) = \{a \}$, $P(b) = \{b \}$, $P(c) = \Omega$, event $E = \{a \}$, we have $\neg K(\{a \}) = \{b,c \}$, $K \neg K(\{a \}) = \{b \}$, therefore $\neg K(\{a \}) \not\subset K \neg K(\{a \})$.}.

Overall, the standard state-space model based on the $K,U$ operators in Eq. \ref{k01} and \ref{u01} cannot incorporate the three well-established properties of necessitation, KU introspection and AU introspection without eliminating all non-trivial unawareness of agents \cite{dekel1998,modica1994,galanis2013,heifetz2006,li2009,schipper2014}.

\section{One approach to preserve non-trivial unawareness} \label{One approach to preserve non-trivial unawareness}

Regard the following Assumption:

\begin{assumption} \label{a01}
Unawareness of event $E \subseteq \Omega$ at the state $\omega \in \Omega$ implies $P(\omega) = \emptyset$.
\end{assumption}

Studies on multiple state-spaces successfully model non-trivial unawareness by specifying subspaces of the full state-space $\Omega$ in which the agent is aware, then requiring that the agent be able to reason about states that are contained in her awareness subspace only. In this approach, the agent is unaware of events at states that are not contained in her awareness subspace\footnote{E.g. Early study of \cite{geanakoplos1989} already refers to full possibility sets as smaller subsets of $\Omega$.}, therefore relaxing the property of necessitation. Relevant studies include the ones of \cite{heifetz2006} on informational lattice of state-spaces, \cite{li2009} on the product-space between the agent's subjective state-space and the factual state-space, and \cite{galanis2013} on multiple knowledge operators conditional on each subspace of $\Omega$. \cite{fukuda2021} provides a detailed analysis on awareness subspaces leading to unawareness of events.

Assumption \ref{a01} derives from the consistent results provided by the multiple state-spaces approach. Unawareness refers to a lack of conception of every state that is possible to obtain \cite{fukuda2021,galanis2013,geanakoplos1989,heifetz2006,li2009,schipper2014}, for if the agent cannot conceive that some state $\omega \in \Omega$ even exists, then she does not have the ability to make a correspondence between this unaware state $\omega$ and any other subset of $\Omega$. It means that the image set of any possibility correspondence $P(\omega)$ regarding an unaware state $\omega$ must be empty, i.e. unawareness at $\omega$ implies $P(\omega) = \emptyset$.

Intuitive nonetheless, Assumption \ref{a01} requires us to revise the standard $K,U$ operators in Eq. \ref{k01} and \ref{u01}, i.e. otherwise, under Assumption \ref{a01} the agent knows all unaware states, $P(\omega) = \emptyset \rightarrow \omega \in KE$ for all $E \subseteq \Omega$, see Eq. \ref{k01}.

The standard $K$ operator in Eq. \ref{k01} derives from the classical work of \cite{aumann1976}. Its original specification from \cite{aumann1976} is equal to $\{ \omega \in \Omega : P_{1}(\omega) \subseteq E \}$, where $P_{1}$ is a space partition satisfying the condition $\omega \in P_{1}(\omega)$ for all $\omega \in \Omega$. Now, the range of a partition $P_{1}$ does not contain the empty set by definition, $\emptyset \not\subset P_{1}$. We propose to regard the original definition of knowledge in \cite{aumann1976} as a reduced version, for which the full specification is equal to the right-hand-side of the equality

\begin{equation} \label{k02}
\begin{array}{rl}
\{ \omega \in \Omega : P_{1}(\omega) \subseteq E \} &= \{ \omega \in \Omega : P_{1}(\omega) \neq \emptyset , P_{1}(\omega) \subseteq E \} \\
\\
&= \{ \omega \in \Omega : \emptyset \neq P_{1}(\omega) \subseteq E \}, \quad P_{1}(\omega) \neq \emptyset . \\
\\
\end{array}
\end{equation}

Eq. \ref{k02} does not hold only for a generalised case with $P(\omega) = \emptyset$ for some state $\omega \in \Omega$. We propose the following Definition\footnote{We assume the notation $\{P(\omega) \neq \emptyset , P(\omega) \subseteq E \} = \{ \emptyset \neq P(\omega) \subseteq E \}$ is unambiguous to the reader.}:

\begin{definition} \label{d01}
For all $E \subseteq \Omega$, $K^{\prime} (E) := \{ \omega \in \Omega : \emptyset \neq P(\omega) \subseteq E \} = K^{\prime} E$. 
\end{definition}

The $K^{\prime}$ operator in Definition \ref{d01} derives from Eq. \ref{k02}, regarding the general possibility correspondence $P : \Omega \rightarrow 2^{\Omega}$. If the image set of $P$ does not include the empty set, then $K^{\prime} E$ is equal to standard $KE$ for all $E \subseteq \Omega$. The usual interpretation applies, for $K^{\prime} E$ is the set of states in which the agent knows $E$, while $\neg K^{\prime} E = \Omega \setminus K^{\prime} E$ is the set of states in which she does not know $E$. Iterations of the complement of the knowledge $\neg K^{\prime}$ as in Definition \ref{d01} derives the unawareness operator satisfying

\begin{equation} \label{u02}
\begin{array}{rl}
U^{\prime}(E) &: \subseteq \neg K^{\prime} E \cap \neg K^{\prime} \neg K^{\prime} E \cap \neg K^{\prime} \neg K^{\prime} \neg K ^{\prime}E \cap \dots \\
\\
&= \bigcap_{i = 1}^{\infty} (\neg K^{\prime})^{i} (E) . \\
\end{array}
\end{equation}

\noindent where $U^{\prime} (E) = U^{\prime} E$ is the set of states in which the agent is unaware of $E$. 

The $K^{\prime}, U^{\prime}$ operators are consistent with the representation of non-trivial unawareness in Assumption \ref{a01}.

\begin{theorem} \label{t01}
$U^{\prime} \Omega = \bigcap_{2^{\Omega}} U^{\prime} E = \bigcup_{\Omega} \{ \omega \in \Omega : P(\omega) = \emptyset \}$.
\end{theorem}
\begin{proof}
Definition \ref{d01} and Eq. \ref{u02}. For any state $\omega \in \Omega$ and any possibility correspondence $P(\omega) = \emptyset$, Definition \ref{d01} implies $\omega \not\in K^{\prime} E$, $\omega \not\in K^{\prime} \neg K^{\prime} E$ for any event $E \subseteq \Omega$. Therefore, Eq. \ref{u02} implies $\omega \in U^{\prime} E$. Since $K^{\prime} E \subseteq K^{\prime} \Omega$ for all events $E \subseteq \Omega$, therefore Eq. \ref{u02} over $\Omega$ implies $\omega \in U^{\prime} \Omega \subseteq U^{\prime} E$ for all events $E \subseteq \Omega$, for all states $\omega$ such that $P(\omega) = \emptyset$.
\end{proof}

Theorem \ref{t01} implies the condition $U^{\prime} \Omega \subseteq U^{\prime} E$ for all events $E \subseteq \Omega$. The set $U^{\prime} \Omega$ is the intersection of all unawareness sets $U^{\prime} E$ with respect to all events $E \subseteq \Omega$, for it represents the non-trivial unawareness that is preserved across the agent's introspection processes. $U^{\prime} \Omega$ satisfies the properties of AU introspection equal to $U^{\prime} \Omega \subseteq U^{\prime} (U^{\prime} \Omega)$, and symmetry equal to $U^{\prime} \Omega = U^{\prime} \emptyset$.

\begin{theorem} \label{t02}
For all $E \subseteq \Omega$, $K^{\prime} E = KE \setminus U^{\prime} \Omega$.
\end{theorem}
\begin{proof}
Definition \ref{d01} and Theorem \ref{t01}. For all events $E \subseteq \Omega$ and any states $\omega, \omega^{\prime} \in E$, Definition \ref{d01} implies $\omega \in KE \rightarrow \omega \in K^{\prime} E$ for all non-empty possibility correspondences $P(\omega) \neq \emptyset$, and $\omega^{\prime} \not\in K^{\prime} E$ for all empty possibility correspondences $P(\omega^{\prime}) = \emptyset$. Theorem \ref{t01} implies $\omega^{\prime} \in U^{\prime} \Omega$ for all $P(\omega^{\prime}) = \emptyset$.
\end{proof}

Theorem \ref{t02} provides a revised version of the property of necessitation equal to $K^{\prime} \Omega = \Omega \setminus U^{\prime} \Omega$, so called 'R necessitation'. It implies the complement $\neg K^{\prime} \Omega = U^{\prime} \Omega$. Definition \ref{d01} and Theorems \ref{t01} and \ref{t02} derive straightforward properties of the $K^{\prime}, U^{\prime}$ operators as follows \cite{bacharach1985,chen2012,dekel1998,fukuda2021,galanis2013,geanakoplos1989,heifetz2006,li2009,schipper2014,tada2024}:

\begin{proposition} \label{p01}
For all $E \subseteq \Omega$, the $K^{\prime}, U^{\prime}$ operators and the unawareness set $U^{\prime} \Omega$ satisfy the properties: \newline

\noindent 1. R necessitation: $K^{\prime} \Omega = \Omega \setminus U^{\prime} \Omega$;

\noindent 2. Monotonicity: $E \subseteq F$ implies $K^{\prime} E \subseteq K^{\prime} F$;

\noindent 3. Truth: $K^{\prime} E \subseteq E$;

\noindent 4. Positive introspection: $K^{\prime} E \subseteq K^{\prime}(K^{\prime} E)$;

\noindent 5. Plausibility: $U^{\prime} \Omega \subseteq U^{\prime} E \subseteq \neg K^{\prime} E \cap \neg K^{\prime} \neg K^{\prime} E$;

\noindent 6. KU introspection: $K^{\prime} (U^{\prime} E) = \emptyset$;

\noindent 7. AU introspection: $U^{\prime} \Omega \subseteq U^{\prime}(U^{\prime} \Omega)$;

\noindent 8. Reverse AU introspection: $U^{\prime} (U^{\prime} E) \subseteq U^{\prime} E$;

\noindent 9. Symmetry: $U^{\prime} \Omega = U^{\prime} \emptyset$.

\end{proposition}

The $K^{\prime}$ operator in Definition \ref{d01} implies the condition $\neg K^{\prime} E \not\subset K^{\prime} \neg K^{\prime} E$, for it violates the property of negative introspection defined as $\neg KE \subseteq K \neg KE$ \cite{chen2012,dekel1998,fukuda2021,tada2024}.

Assume an example with a state-space $\Omega = \{a,b,c,d \}$ and the possibility correspondence $P(a) = \{a \}$, $P(b) = \{b \}$, $P(c) = \emptyset$, $P(d) = \Omega$. Apply the $K^{\prime},U^{\prime}$ operators in Definition \ref{d01} and Eq. \ref{u02}. The agent's non-trivial unawareness set is equal to $U^{\prime} \Omega = \{c \}$, see Theorem \ref{t01}. The agent knows the event $E = \{a \}$ at the state $K^{\prime} (\{a \}) = \{a \}$. By applying the $U^{\prime}$ operator in Eq. \ref{u02} over $E = \{a \}$, we find $U^{\prime}(\{ a \}) \subseteq \{c,d \}$, which is plausible with her unawareness set, implying $U^{\prime} \Omega = \{c \} \subseteq U^{\prime}(\{ a \}) \subseteq \{c,d \}$, see Proposition \ref{p01}. For the full event $F = \Omega$, the agent knows $F$ at the states $K^{\prime} F = K^{\prime} \Omega = \{ a,b,d \}$ by direct application of $K^{\prime}$ over $\Omega$, while she is unaware of $F$ at the state $U^{\prime} F = U^{\prime} \Omega = \{ c \}$. The non-trivial unawareness set $U^{\prime} \Omega = \{c \}$ indicates that at the state $c$, the agent does not know whether event $\{a,b,c,d \}$ or $\{a,b,d \}$ obtains. Indeed, this is the case for any event which includes some unaware state $\omega \in U^{\prime} \Omega$, for it implies a general condition $K^{\prime} (E \cup U^{\prime} \Omega) = K^{\prime} E$ for all $E \subseteq \Omega$, see Theorem \ref{t02}.

The property of R necessitation in Proposition \ref{p01} may provide a solution for the DLR impossibility result within the standard state-space model. Apply the $K^{\prime}, U^{\prime}$ operators over the derivation in \cite{dekel1998}, see Eq. \ref{dlr01}, and assume non-trivial unawareness $U^{\prime} \Omega \neq \emptyset$. We find

\begin{equation} \label{dlr02}
\begin{array}{rl}
\emptyset \neq U^{\prime} \Omega &\subseteq U^{\prime} (U^{\prime} \Omega) \text{ (AU introspection)} \\
\\
&\subseteq \neg K^{\prime} \neg K^{\prime} (U^{\prime} E) \text{ (Eq. \ref{u02}, Plausibility)} \\
\\
&= \neg K^{\prime} \Omega \text{ (KU introspection)} \\
\\
&= U^{\prime} \Omega , \text{ (R necessitation)} \\
\\
\end{array}
\end{equation}

\noindent therefore preserving the agent's non-trivial unawareness $U^{\prime} \Omega \neq \emptyset$.

\section{Discussion} \label{Discussion}

Assumption \ref{a01} regards empty possibility correspondences to represent non-trivial unawareness of agents, for the $K^{\prime}$ operator in Definition \ref{d01} adjusts the standard knowledge operator $K$ to be consistent with this representation. We may interpret that our approach in Section \ref{One approach to preserve non-trivial unawareness} derives two sets of unaware states, the core set $U^{\prime} \Omega$ in Theorem \ref{t01} referring to the non-trivial unawareness that is preserved through the agent's introspection processes, see Eq. \ref{dlr02}, and the more comprehensive set $U^{\prime} E$ defined in Eq. \ref{u02}. From plausibility $U^{\prime} \Omega \subseteq U^{\prime} E$ in Proposition \ref{p01}, Eq. \ref{u02}, the complement set $U^{\prime} E \setminus U^{\prime} \Omega \neq \emptyset$ refers to the non-trivial unawareness that might be resolved through the agent's own introspection, which derives from the non-partitional approach proposed by \cite{geanakoplos1989}. 

It implies two cases. On the one hand, studies show that the comprehensive set $U^{\prime} E$ satisfying $U^{\prime} E \setminus U^{\prime} \Omega \neq \emptyset$ violates the property of AU introspection, see \cite{chen2012,fukuda2021,tada2024}, therefore the DLR impossibility result as derived in Eq. \ref{dlr01} does not hold. On the other hand, the core set $U^{\prime} \Omega$ in Theorem \ref{t01} satisfies AU introspection, although this set $U^{\prime} \Omega \neq \emptyset$ does not produce a contradiction, see Eq. \ref{dlr02}. Hence, this approach preserves the agent's non-trivial unawareness either by the set $U^{\prime} E \neq \emptyset$ not holding to the derivation in Eq. \ref{dlr01} \cite{chen2012,fukuda2021,tada2024}, or by the set $U^{\prime} \Omega \neq \emptyset$ not leading to any contradiction in Eq. \ref{dlr02}. It may be an explanation and a possible solution for the classical DLR impossibility result.



\end{document}